\documentclass[a4paper,twocolumn,11pt, accepted=2023-07-16]{quantumarticle}
\pdfoutput=1
\usepackage[utf8]{inputenc}
\usepackage[english]{babel}
\usepackage[T1]{fontenc}
\usepackage{amsmath}
\usepackage{amsthm}
\usepackage{dsfont}
\usepackage{hyperref}
\usepackage{physics}
\usepackage{graphicx}
\usepackage{tikz}
\usepackage{lipsum}
\usepackage{dsfont}
\usepackage{float}
\usepackage{multicol}

\usepackage{thm-restate}
\declaretheorem[name=Theorem]{thm}
\declaretheorem[name=Lemma]{lemm}

\newcommand{\win}{\text{win}}

\begin{document}

\title{A coherence-witnessing game and applications to semi-device-independent quantum key distribution}

\author{M\'{a}rio Silva}
\affiliation{Université de Lorraine, CNRS, Inria, LORIA, F-54000 Nancy, France}
\email{mmachado@loria.fr}
\author{Ricardo Faleiro}
\email{ricardofaleiro@tecnico.ulisboa.pt}
\affiliation{Instituto de Telecomunicações, 1049-001, Lisbon, Portugal}
\author{Paulo Mateus}
\affiliation{Instituto de Telecomunicações, 1049-001, Lisbon, Portugal}
\affiliation{Departamento de Matemática, Instituto Superior Técnico, Avenida Rovisco Pais 1049-001, Lisbon, Portugal}
\author{Emmanuel Zambrini Cruzeiro}
\affiliation{Instituto de Telecomunicações, 1049-001, Lisbon, Portugal}
\email{\\emmanuel.zambrinicruzeiro@gmail.com}
\maketitle

\begin{abstract}
Semi-device-independent quantum key distribution aims to achieve a balance between the highest level of security, device independence, and experimental feasibility. Semi-quantum key distribution presents an intriguing approach that seeks to minimize users' reliance on quantum operations while maintaining security, thus enabling the development of simplified and hardware fault-tolerant quantum protocols. In this work, we introduce a coherence-based, semi-device-independent, semi-quantum key distribution protocol built upon a noise-robust version of a coherence equality game that witnesses various types of coherence. Security is proven in the bounded quantum storage model, requiring users to implement only classical operations, specifically fixed-basis detections.
\end{abstract}

\section{Introduction}

Modern cryptosystems based on the hardness of computational assumptions are  vulnerable to developments of computational power and new algorithms, particularly when considering quantum computers \cite{quantumcrypto,Shor}. Quantum cryptography offers a solution to this problem by providing information security based on the laws of quantum mechanics, so that protocols are resistant to any attack, no matter how much  computational power is allowed.

However, in a practical scenario, the users are usually incapable of verifying that their devices follow the description of the protocol and must trust their manufacturer. Not only could it be difficult to create devices that perfectly satisfy the assumptions of the protocol, but it could also be the case that the devices have been maliciously constructed. Quantum cryptography as originally introduced, e.g. BB84 \cite{BB84}, had been predicated on assumptions about the internal description of the physical systems, such as the source and detectors, which opens the door to various side channel attacks. However, in 1998, Mayers and Yao put forth the concept of self-testing \cite{mayers1998quantum,MY04}, ensuring that, if certain statistical tests are met, then the source could be guaranteed to satisfy the desired assumptions, e.g. for quantum key distribution (QKD). Thanks to this property, two users may certify on their own whether their apparatus are functioning as they should. This is the key idea behind \textit{device-independent} (DI) QKD, which aims at unconditional security in the presence of imperfect, or maliciously designed, devices.  

DI QKD is the golden standard of QKD: it allows unconditional security based on the laws of physics even for untrusted or maliciously designed devices \cite{Vazirani_2014,Arnon-Friedman2018}. Device-independence also finds other applications in cryptography: random number generation \cite{Pironio_2010_Random, AMPS12}, coin flipping \cite{AC11}, and authorization to private databases \cite{FG}. DI QKD, for now, remains extremely challenging. The first proof-of-principle experiments were performed only very recently \cite{Nadlinger2022,Zhang2022,Liu2022}, almost 40 years after the invention of BB84. 

It then becomes naturally interesting to study scenarios which may reach a compromise between experimental challenge and security: for example, by assuming than the users have a partial description of the devices --- say, one device is trusted while the other is not. These are called \textit{semi-device-independent} (SDI) protocols. Existing approaches include: bounding the dimension of the states \cite{Pawlowski_Brunner_2011, Chaturvedi2018,  Tavakoli_2018, Tavakoli_2020}, bounding the expectation value of some appropriate operator, e.g energy \cite{VanHimbeeck2017}, bounding the information content of quantum states \cite{Tavakoli2020info,Tavakoli2022info}, or their overlap \cite{Shi2019}.

Another interesting question is what aspects of a protocol must be strictly quantum in order to guarantee security through the laws of quantum mechanics. \textit{Semi-quantum} (SQ) cryptography \cite{iqbal2019semiquantum} attempts to answer this question. There, one is interested in minimizing the quantum technological requirements of the systems and/or users involved in the protocol without compromising security. One way to do this is to limit Alice or Bob to a single measurement basis, for instance, or force them to only perform detection or reflection of photons \cite{boyerSQKD,Massa2022experimentalsemi}.

In this work, we take a first step towards the intersection of semi-device-independent and semi-quantum protocols. The security proof for our QKD protocol relies on specific properties of the detection operators used by Alice and Bob in their respective labs---modelled as simple single-basis measurements---, but the source and measurements controlled by the outside servers remain unspecified. 

The protocol is based on a generalization of the \textit{Coherence Equality} (CE)  game  introduced by del Santo and Dakić \cite{Del_Santo_2020} to a noise-robust version. Its security follows the standard approach of DI QKD proofs \cite{Pironio_2010_Random}, and is established from the gap between the optimal quantum and classical performance bounds in the game.

The paper is organized as follows. In Section \ref{sec:two}, we introduce the scenario and describe the quantum correlations it exhibits, namely we determine the optimal quantum bounds using lower and upper bounding techniques \cite{sdp}. We further interpret the game as a coherence witness, and discuss the role of randomness in the game.  In Section \ref{sec:three}, we introduce the SDI SQKD protocol and prove the security of the protocol in the bounded quantum storage model. Finally, we present our conclusions in Section \ref{sec:four}.

\section{The Coherence Equality game}
\label{sec:two}

\subsection{Scenario and basic definitions}
\label{subsec:scenario}

The scenario is depicted in Fig.\ref{fig:setup_QKD}.  An untrusted source sends a (quantum) state $\rho_\text{AB}$ to Alice and Bob's labs. Alice and Bob are allowed to perform a local single basis measurement on their side of the system, where they either block the path possibly taken by the particle, or leave it undisturbed. Their actions are governed by random bits \(x\) and \(y\), respectively for Alice and Bob, such that \(0\) corresponds to leaving the path undisturbed and \(1\) to blocking it.
In the case where they block the path, they use a single photon detector to determine whether any particles were  present on their side, or not. This information is given by output bits $\alpha, \beta$, for Alice and Bob, respectively, --- where \(0\) corresponds to no particle detected and \(1\) to at least one particle detected. Finally, Alice and Bob either send their part of the initial system undisturbed, or a vacuum state produced from blocking, to the untrusted servers $S_A, S_B$, respectively, where a generalized quantum measurement can be performed. The final measurements then produce outcomes $a,b$, for Alice, Bob, respectively, and  they win the game if \(a\oplus b = x \oplus y\). The statistics of such an experiment can be described by a probability distribution $p(ab|xy)$, and we consider the following linear functional of the probabilities, giving the winning probability of the game
\begin{equation}\label{eq:witness}
    P_\text{win}:=\frac{1}{4}\sum_{a,b,x,y}\mathbf{1}_{a\oplus b=x\oplus y}\,p(ab|xy),
\end{equation} 
where $\mathbf{1}_F$ is 1 whenever \(F\) holds and 0 otherwise. \\

   \begin{figure*}[t]
    \centering
    \includegraphics[trim={5mm 0cm 5mm 0cm},clip,width=0.9\textwidth]{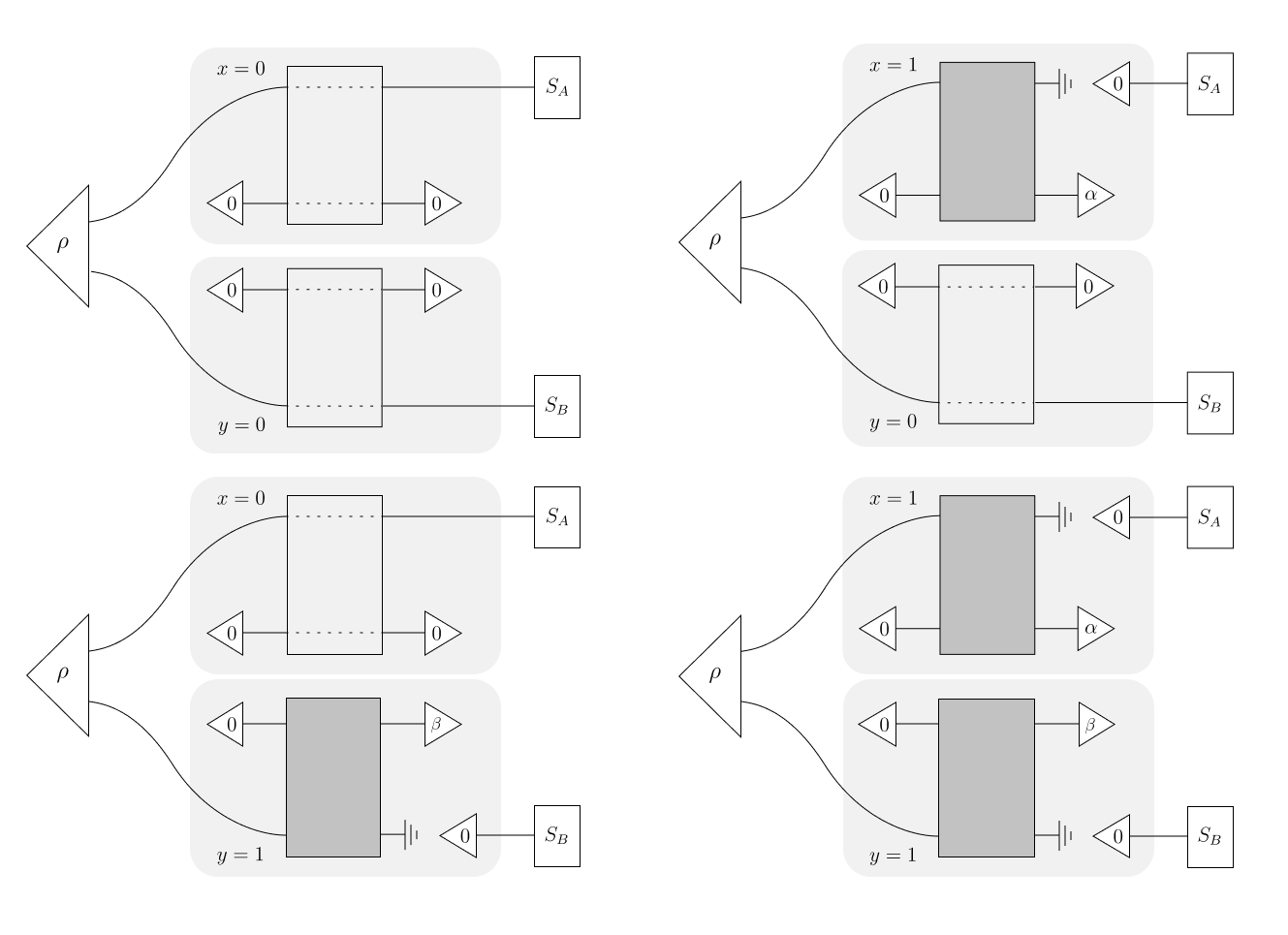}
    \caption{Diagram of the Coherence Equality (CE) game, for each of the four equally probable detection configurations:  Alice and Bob do not try to detect, for \(x=y=0\) (top-left); Alice tries to detect and Bob does not, for \(x=1\) and \(y=0\) (top-right); Bob tries to detect and Alice does not, for \(x=0\) and \(y=1\) (bottom-left);  Alice and Bob both try to detect, for \(x=y=1\) (bottom-right). The detections occur inside the trusted labs of each player (light gray regions), and are represented by the blocks in dark gray. The detection process has two input states: the unknown state from the source and an ancillary input state, always initialized at \(0\); and a single output given by the classical bit (\(\alpha\) for Alice and \(\beta\) for Bob), informing if the detection was successful or unsuccessful. When the detection is attempted, whether it may be successful or not, the quantum system is blocked from its original path (represented by the \textit{ground symbol}) and a vacuum state \(\ketbra{0}\) is communicated to the servers instead. In the case where no detection is attempted, the original state proceeds undisturbed through the  lab (represented by a dashed line) to the untrusted servers \(S_A,S_B\), and the classical system remains unchanged. The servers \(S_A,S_B\) are constrained by the no-signalling condition which forbids them from communicating their final measurement results, i.e  bits \(a,b\), but can have shared randomness or entanglement.}
    \label{fig:setup_QKD}
\end{figure*}

There are some noticeable differences between our setup and the original setup of del Santo and Dakić \cite{Del_Santo_2020}. To start, in the original setup Alice and Bob perform the final measurements, whereas the blocking is done by some neutral Referees. In our case, since we want to limit the powers of Alice and Bob, we have placed them in the role of randomly blocking the channels, whereas some untrusted servers $S_A$ and $S_B$ can perform the final measurements. Furthermore, Alice and Bob  receive one bit of information each when blocking/detecting. Finally, and more crucially, in \cite{Del_Santo_2020} it was assumed that the source was verified and only outputted a single photon, meaning that there was exactly a \textit{single photon shared between the labs}. Instead, here, we consider the possibility of an untrusted source outputting any number of photons, 
thus allowing \textit{more than one photon per and between labs}. Alice and Bob will constrain these  multi-photon correlations that may arise through a statistical constraint, called the \textit{single-detection constraint}. This constraint can be operationally verified via their detection results, when both attempt a detection inside their lab.

\subsection{The game in two stages}

The game can naturally be separated into two stages.
In the first stage,  Alice or Bob may choose to detect their particle. They use single-particle detectors that only  need to reliably distinguish between vacuum and non-vacuum states. These measurements can be represented by the detection operators $D_{\alpha|x},D_{\beta|y}$ where, for $x,y=1:$
\begin{equation}\label{eq:detectorproperties}
 \left \{  D_{0|1} = \ketbra{0};\; D_{1|1} = \sum_{i=1}^{i=d} \ketbra{i}\right \}
\end{equation}
where $d$ corresponds to the number of photons allowed per lab, which implies a \(d+1\) dimensional system in each lab, that is,  taking into account the vacuum state.  It should be noted that since the measurements 
 are destructive, whenever a detection is attempted the post-measurement state will be always set to \(\ketbra{0}\), regardless of the particular outcome observed  (see Fig. \ref{fig:setup_QKD}). 

When the users choose not to block ($x,y=0$), they simply act with the appropriately normalized identity operator on their subsystems,
\begin{equation}
    D_{0|0} = \frac{\mathds{1}}{d+1}
\end{equation}

Therefore, after passing through Alice's and Bob's labs, for inputs $x,y$, the original state $\rho_{AB}$ is transformed into $\rho_{xy}$, for
 \(\rho_{00} = \rho_{AB};\;  \rho_{01} = \Tr_B{(\rho_{AB})}\otimes\ketbra{0}_B;\;
 \rho_{10} = \ketbra{0}_A\otimes \Tr_A{(\rho_{AB})};\;
 \rho_{11} =\ketbra{0}_A\otimes \ketbra{0}_B\). In this way, the coherence equality game can alternatively be understood as a game played by the servers, where they receive quantum inputs instead of classical ones, which accounts for the difficulty of computing the equality \(a\oplus b=x\oplus y\), otherwise trivial for classical inputs.
    
The probability distribution describing the statistics for this first step of the protocol via the Born rule, is defined as
\begin{equation}
    \tilde{p}(\alpha\beta|xy) = \text{Tr}[\rho D_{\alpha|x}\otimes D_{\beta|y}].
\end{equation}

Furthermore, the \textit{single-detection} condition will appropriately constrain some of the probability distributions, namely it demands that, 
\begin{equation}
\label{d_eps}
    \tilde{p}(1,1|1,1) \leq d_\epsilon.
\end{equation}

That is, the probability of both Alice and Bob detecting a non-vacuum state in each of their labs is bounded by \(d_\epsilon\). Thus, in our scenario, we consider general multi-photon states by replacing the single-particle condition with a single-detection 
condition, which can be verified by the users.

In the second stage of the game the state $\rho_{xy}$ is forwarded to the external servers \(S_A\) and \(S_B\), where each server may apply a Positive Operator-Valued Measure (POVM), which is a set of positive semidefinite operators $\{A_{a}\}_{a=1}^{n_A}$ for Alice's server, and $\{B_{b}\}_{b=1}^{n_B}$ for Bob's, that sum to the identity.
\begin{align}
    & A_a \geq 0,\quad \forall a,\qquad \sum_{a=1}^{n_A} A_a = \mathds{1},\\
    & B_b \geq 0,\quad \forall b,\qquad \sum_{b=1}^{n_B} B_b = \mathds{1}
\end{align}
where $n_A, n_B$ are the number of possible outcomes for each POVM, respectively.

Finally, for this stage the probability distribution is computed via the Born rule as,
\begin{equation}
    p(ab|xy) = \text{Tr}[\rho_{xy}A_a\otimes B_b]
\end{equation}

The problem we wish to solve is to maximize the linear functional Eq.~(\ref{eq:witness}) over the states and POVMs. This represents  an instance of Semidefinite Programming (SDP), and can therefore be efficiently solved for a fixed dimension $d$.

In the following paragraph we describe an equivalent formulation of the single detection constraint,  which will provide a useful perspective for the calculation of the quantum bounds.  

Regarding the constraint imposed by Eq. (\ref{d_eps}), one can alternatively consider that the above condition  imposes a direct constraint on the elements of the state $\rho^{(d)}$, describing \(d-\)photons per lab, instead of restricting the first-stage operational statistics. This will be fruitful to optimize the quantum model (over states and measurements) in order to find the maximum winning probability of the CE game. According to this perspective, we omit from the description the outcomes $\alpha$, $\beta$, and merely use the blocking operators to create the four states $\rho_{xy}$, from a state $\rho^{(d)}$ having some appropriate matrix elements bounded. The elements we wish to bound in $\rho^{(d)}$ are $\{\rho_{(i,i)}|\; i \in  \mathcal{I}_d \}$,  where \( \mathcal{I}_d\) is a set that specifies the amplitude elements corresponding to more than one excitation per party for the \(d\) photon case, i.e.
\begin{equation*}
    \mathcal{I}_d =  \left \{\left[\alpha (d +1) +2, (\alpha+1) (d+1) \right] |\;  \alpha \in [d] \right. \}.
\end{equation*}

The single-detection constraint will then bound the sum of those elements, i.e.
\begin{equation}
\label{state constrained}
    \sum_{i \in \mathcal{I}_d } \rho_{i,i} \leq d_\epsilon.
\end{equation}That is, it demands that the total probability of having simultaneously non-vacuum states in each lab is bounded.  It is easy to show that for a generic two-qubit state (i.e. \(d=1\)), this constraint yields \(\rho_{(4,4)} \leq d_\epsilon\), where \(\rho_{(4,4)}\) is the element  corresponding to the state \( \ketbra{11}\). Furthermore, if we demand that \(d_{\epsilon}= 0\) then this leads to \(\rho_{(4,4)} = 0\), which recovers the single-particle condition.

\subsection{Boundary of the quantum set}

The most general bi-partite pure state in an \(n \otimes n\) dimensional Hilbert space is given by
\begin{multline}
  \ket{\Psi_n} = C_{00}\ket{00}+\sum_{i\in [n]} C_{i0} \ket{i0}+\\ \sum_{j\in [n]} C_{0j} \ket{0j}+\sum_{i,j\in [n]\times[n]} C_{ij} \ket{ij}.
\end{multline}

It is straightforward to see that the ideal strategy for coherence equality game requires one photon per lab, since it is the most economical way to guarantee that Alice and Bob win with certainty. In fact, having more than one photon in each 
  lab is inconsequential for the coherence equality game --- say, either Alice and Bob have one photon each, in which case they always win and thus dispense the need for more photons, or if Alice does not have any photons in her lab, then regardless of whether Bob has just a single photon or more they can only win half the time since $S_A$'s best strategy will always be a random guess. This suggests that, in the quantum case one can  consider a simpler state than the previous general state as an equivalent resource, obtained by assuming that any  \(\ket{n}\) for \(n \geq 1\) can be mapped to \(\ket{1}\), yielding a general two-qubit pure state,
\begin{equation}
\label{two qubit state}
    |\psi\rangle = c_{00}|00\rangle + c_{01}|01\rangle + c_{10}|10\rangle + c_{11}|11\rangle
\end{equation} where $|c_{00}|^2+|c_{01}|^2+|c_{10}|^2+|c_{11}|^2 = 1$, and \(|c_{00}|^2=|C_{00}|^2, |c_{10}|^2 = \sum_{i} |C_{i0}|^2,  |c_{10}|^2 = \sum_{j} |C_{0j}|^2,\) and \( |c_{11}|^2 = \sum_{i,j} |C_{ij}|^2\). Indeed, qubits are optimal in our problem, thus there is no loss of generality by considering Eq. (\ref{two qubit state}), something we verified this numerically for small dimensions. 

We intend to maximise Eq. (\ref{eq:witness}). First, we note that real states and measurements are sufficient for optimality, as in the case of Bell inequalities \cite{Pal2008,Mckague2009}.

Additionally, note that for fixed $d_\epsilon$, the strongest correlations can be achieved with a state that saturates the constraint $\rho_{44} \leq d_\epsilon$. Indeed, reducing $\rho_{44}$ can only lead to weaker correlations. It is straightforward to note this since for the problem at hand $|11\rangle$ is optimal, i.e. achieves 100$\%$ success probability.

Furthermore, we assume \(S_A\) and \(S_B\) use the same measurement operators\footnote{Note this marks a strong difference between the CE game and Bell nonlocal games, where in the latter non-commuting measurements are typically required for optimality of a quantum behaviour.}, described by the following equatorial projections on the Bloch sphere
\begin{equation}
    M_a = \frac{1}{2}\Big(\mathds{1} + (-1)^a
     (n_x X+ \sqrt{1 - n_x^2} Z)\Big), 
\end{equation}where we have used the fact that \(n_z = \sqrt{1 - n_x^2}\), since for equatorial projections \(n_y=0\). Finally, we assume that $c_{01} = c_{10}$, since the game is symmetric with respect to both the players' and servers' operations and outputs.

We have verified that all the above assumptions give rise to the optimal solution using the SDP solver SDPT3 \cite{Toh1999}. Taking all of them into account, the simplest form of an optimal state to play the CE game is given by 
\begin{equation}
       c_{00}|00\rangle + c_{01}(|01\rangle +|10\rangle)+ \sqrt{d_{\epsilon}}|11\rangle.
\end{equation}

Using the simplified form of the state and of the measurement, the optimization problem reduces to
\begin{align}\label{eq:succ_simplified}
    P_\text{win} & = \frac{1}{8}\Big[ 2c_{00}\sqrt{d_\epsilon}n_x^2 - 8c_{01}\sqrt{d_\epsilon}n_x \sqrt{1-n_x^2} \nonumber\\
    & - (1+3d_\epsilon)(n_x^2-2) + 4c_{01}^2(1+n_x^2)\nonumber\\
    & + c_{00}^2 (2+n_x^2)\Big]
\end{align}
subject to
\begin{equation}\label{eq:succ_constraint}
    c_{00}^2 + 2c_{01}^2 = 1 - d_\epsilon.
\end{equation}

The problem can be readily converted into a polynomial optimization problem, which can then be efficiently solved using lower and upper bounding procedures, thereby ensuring convergence to a global maximum. We use an SDP see-saw \cite{Werner2001} for the lower bounds, and YALMIP \cite{Lofberg2004} for the upper bounds, to show optimality of our results. The boundary of the quantum set is given by the full line shown in Fig. \ref{fig:Qbounds}.

\begin{figure}[h!]
    \centering
    \includegraphics[width=1\columnwidth]{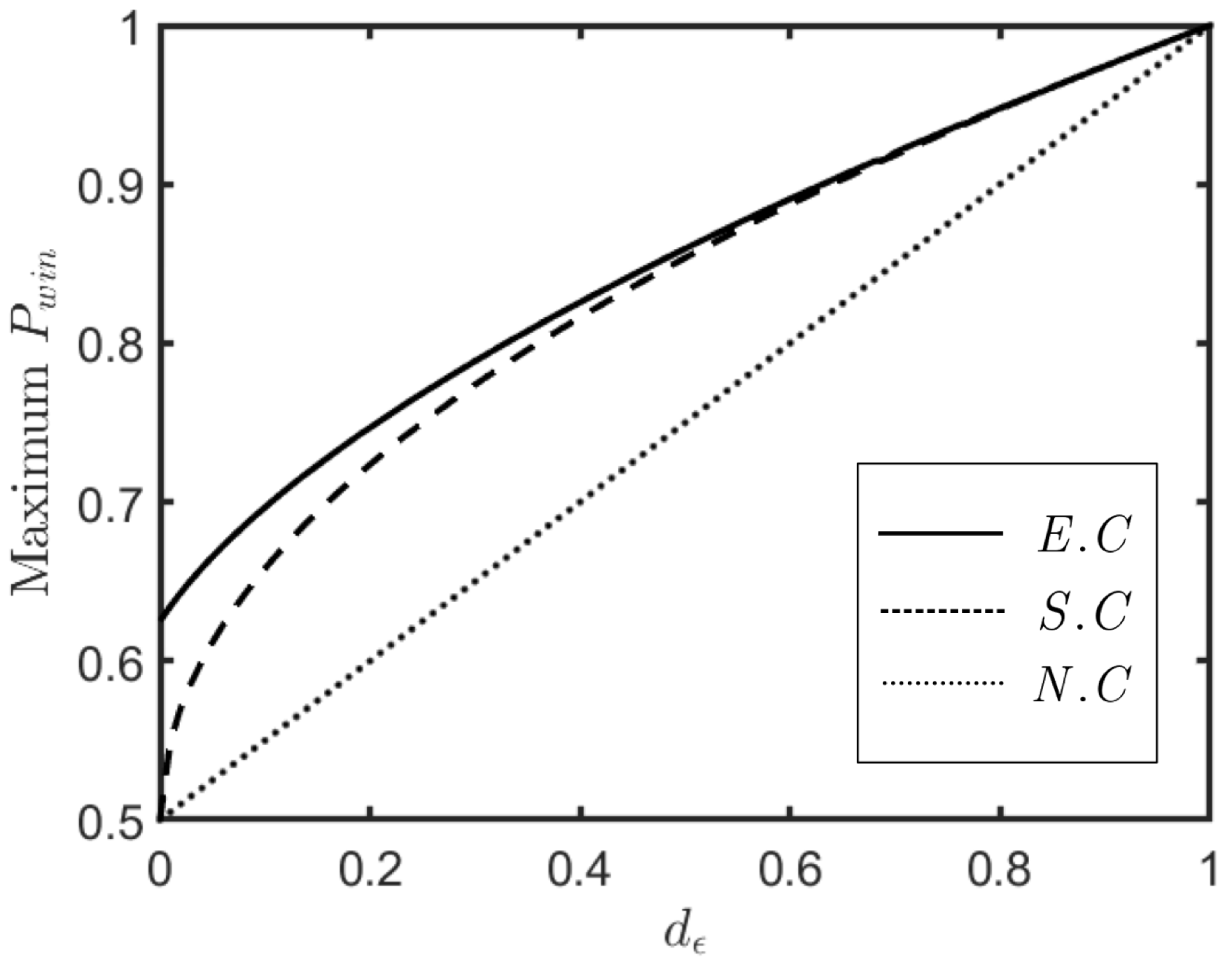}
    \caption{Optimal bounds for the winning probability given by different types of coherent resources. The dotted line corresponds to non-coherent (\textit{N.C}) states. The dashed line is the optimal bound for separable coherent (\textit{S.C}) states. Finally, the full line shows the overall optimal bound, attained using entangled coherent (\textit{E.C}) states.}
\label{fig:Qbounds}
\end{figure}

\subsection{Coherence witness and analytic solutions}

Through the introduction of the \(d_{\epsilon}\) parameter, we are able to identify three distinct classes of coherence resources which form a strict hierarchy in terms of their performance in the CE game. The three classes are:
\begin{enumerate}
    \item Entangled coherent states:
    \begin{multline}
    \label{entcoherent}
        \sqrt{x(d_{\epsilon})}\ket{00} +{\sqrt{d_{\epsilon}}\ket{11}} \\+ \sqrt{\frac{1-d_{\epsilon}-x(d_{\epsilon})}{2}}\left(\ket{01}+\ket{10}\right).
         \end{multline}
    \item Separable coherent states:
    \begin{equation}
          ( \sqrt{1-{d_{\epsilon}}}\ket{0}+\sqrt{d_{\epsilon}}\ket{1}){\otimes \ket{1}}.
     \end{equation}
    \item  Mixed non-coherent states:
    \begin{equation}
    ({1-d_{\epsilon}})\ketbra{10}+d_{\epsilon}\ketbra{11}.
    \end{equation}
  
\end{enumerate}
For the entangled coherent states (Eq.\ref{entcoherent}), \(x(d_{\epsilon})\) represents the analytic expression of a function of \(d_{\epsilon}\) obtained by solving the equations in Appendix~\ref{AppendixA}.

In this way, the game can also be interpreted as coherence witness, allowing the certification between three  different types of coherence resources, as given by Alice and Bob's ability to win the game and  \(d_{\epsilon}\). Although the type of coherence witnessed is not basis-invariant \cite{set_coherence, wagner2023inequalities}, as it pertains to a basis of occupation modes and is not invariant under unitaries, it may still be of interest since a description in terms of occupation modes is naturally adopted when analysing quantum optics experiments. 

It is worthy to point out that for \(d_{\epsilon}=0\), separable coherent and mixed non-coherent states reduce to pure non-coherent states, meaning that the previous resources can only be defined for positive values of \(d_{\epsilon}\), and that only entangled coherent states and pure non-coherent states were considered in the original model \cite{Del_Santo_2020} i.e. for \(d_{\epsilon}=0\). In Fig. \ref{fig:VN} we show how the entanglement of the optimal coherent states behaves as a function of  \(d_{\epsilon}\). The entanglement is maximum for \(d_{\epsilon}=0\), corresponding to the state \(\ket{\psi}= (\ket{01}+\ket{10})/{\sqrt{2}}\), and reaches zero for \(d_{\epsilon}=1\) corresponding to the separable state \(\ket{11}\).

\begin{figure}[h!]
    \centering
    \includegraphics[width=1\columnwidth]{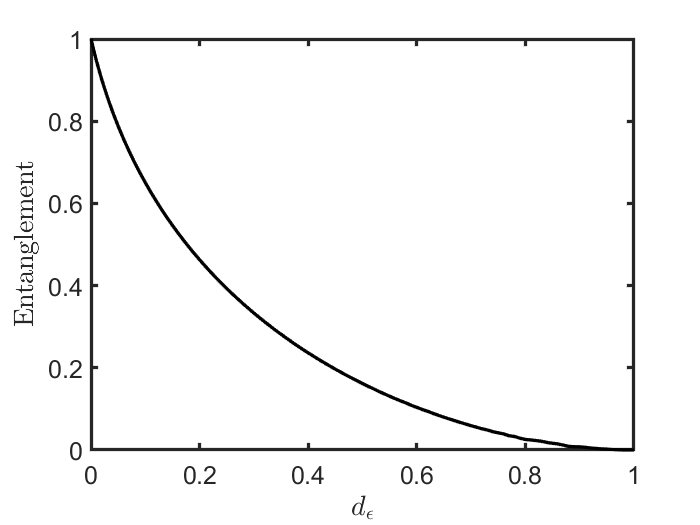}
    \caption{The entropy of entanglement, measured in bits, of the optimal entangled coherent states that reach the maximum winning probability (the solid line  in Figure 2), as a function of \(d_{\epsilon}\). }
\label{fig:VN}
\end{figure}

Regarding the separable coherent states and mixed non-coherent states, trivially they will have no entanglement for any value of \(d_{\epsilon}\). As such, only the entangled coherent states can give rise to genuine randomness, and only these will be useful for the purposes of the  QKD protocol we present next. We show this explicitly in Fig. \ref{fig:Hmin} by plotting the relationship between \(H_{min}\) as a function of the winning probability, for various values of \(d_{\epsilon}\).

\subsection{Randomness in the CE game}
\label{section:random}

As we have seen, for certain values of \(d_\epsilon\) and \(P_\text{win}\) the game can certify that the source is sending a particle in a superposition to Alice and Bob's labs, that is, an entangled coherent state is being used as a resource. This suggests that, in this case,  Alice and Bob's detection outcomes may have some randomness. In order to verify this, we employ SDP techniques to maximize Alice's (or alternatively Bob's) marginal probability when she attempts a detection, i.e. \(p(\alpha|x=1)\), as a function of the winning probability \(P_\text{win}\) and \(d_\varepsilon\). The maximum of the marginal will also be an upper-bound on the guessing probability of an adversary, which we refer to as \(P_g(P_\text{win},d_{\epsilon})\). The bounds are presented in Fig.\ref{fig:Pguess}.

\begin{figure}[h!]
\centering\includegraphics[width=1\columnwidth]{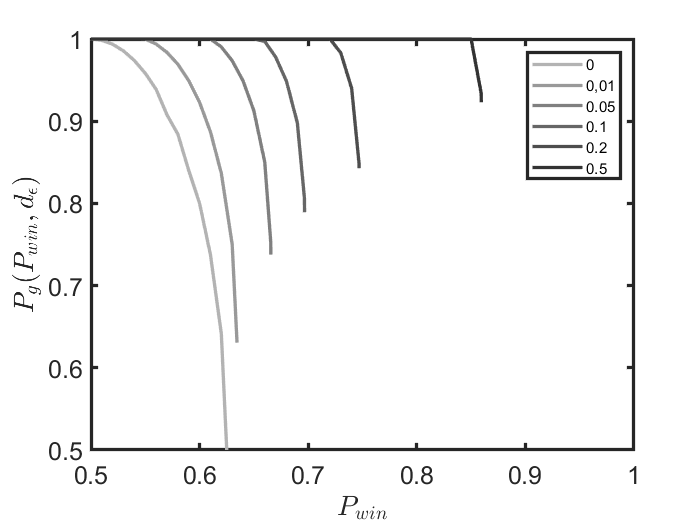}
    \caption{Upper bound on the guessing probability of Alice's outcome, as a function of \(P_{\textup{win}}\) up to \(\textup{Max}(P_{\textup{win}}(d_{\epsilon}))\), for   \(d_{\epsilon}\) taking values  \(0,0.01,0.05,0.1,0.2,\) and \(0.5\), from left to right.}
\label{fig:Pguess}
\end{figure}

We verify that for values of the winning probability that surpass the separable-coherent bound, at which point one requires entanglement, the  detection of a particle by Alice becomes random, and the greater the gap from the separable-coherent bound the less predictable it is. Also, this relationship is stronger for smaller  values of the noise-parameter $d_{\epsilon}$. Alternatively this bound can be represented by the min-entropy,  $H_{\min}(P_{\textup{win}},d_\epsilon)=-\log_2(P_{g}(P_\text{win},d_\epsilon))$, and is plotted in Figure \ref{fig:Hmin}.

\begin{figure}[h!]
    \centering
    \includegraphics[width=1\columnwidth]{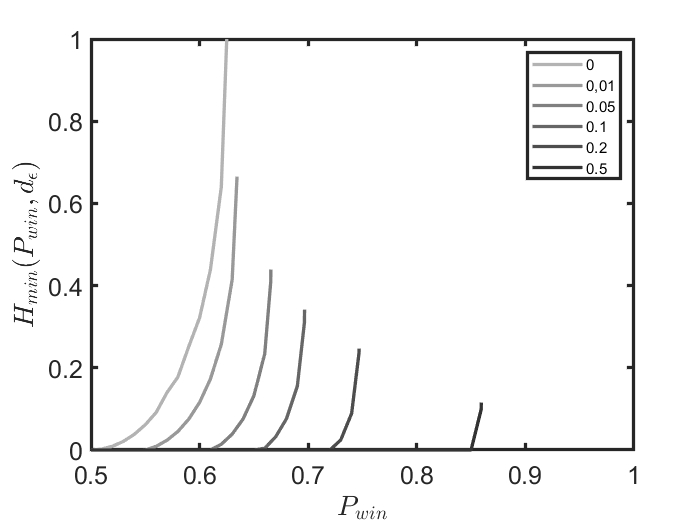}
    \caption{Randomness shared between Alice and Bob, measured in bits by the min-entropy (\(H_{\min}\)), as a function of \(P_{\textup{win}}\) up to \(\textup{Max}(P_{\textup{win}}(d_{\epsilon}))\), for   \(d_{\epsilon}\) taking values  \(0,0.01,0.05,0.1,0.2,\) and \(0.5\), from left to right. }
\label{fig:Hmin}
\end{figure}

 The min-entropy takes positive values precisely when the probability surpasses the optimal bound achievable by the separable coherent resources. 
These results highlight an important feature of the game, that a set of state and measurements $(\rho,M_{A}^a,M_{B}^b)$ picked by some adversary cannot simultaneously be used to win the CE game arbitrarily well and allow the outcomes $\alpha$ and $\beta$ to be completely predictable. We capture this feature of the game in the following proposition.

\begin{restatable*}{prop}{lemmasdpNoGo}
\label{lemma:sdpNoGo}
Consider $P_{g}(P_\text{win},d_\epsilon)$, the upper-bound on an adversary guessing probability of the players' detection outcomes. Then, for any given round $i$ in the protocol, the following three conditions cannot simultaneously hold:

\begin{enumerate}

    \item The no-signalling condition is satisfied and our assumptions for  Alice's and Bob's detectors are correct.

    \item The set \(S_i=(\rho,M_{A}^a,M_{B}^b)_i\) of the state and measurement operators picked by an adversary wins the CE game with probability at least \(P_\text{win}\), and satisfies the single-detection condition with probability at least \(1-d_\epsilon\). 
    
    \item There is an outcome \(\alpha'\in\{0,1\}\)  such that \(p(\alpha=\alpha'|x=1)>P_{g}(P_\text{win},d_\epsilon)\).
    
\end{enumerate}
\end{restatable*}
\begin{proof}
The first condition is assumed to hold in a correct implementation of the setup, under which the incompatibility between the second and third conditions comes directly from the bounds on the detection probability (see Fig.\ref{fig:Pguess} and Fig.\ref{fig:Hmin}).
\end{proof}

\section{QKD protocol}
\label{sec:three}

The protocol consists of \(m\) rounds (with indices in \([m]\equiv\{1,\dots,m\}\)), where Alice and Bob uniformly and randomly choose their respective inputs, \(x,y\in\{0,1\}\). The set \(\mathbf{D}\) consists of the indices of the rounds where both Alice and Bob chose to attempt detection, i.e. \(x=y=1\). To perform the step of classical post-processing, Alice chooses a random subset \(\mathbf{B}\) of \(\mathbf{D}\) with size \(\gamma |\mathbf{D}|\), where \(\gamma>0\) is small. With the outcome information of set \(\mathbf{B}\), Alice and Bob estimate the single-detection parameter \(d_\epsilon\) (Equation \ref{eq:singledetectionCE}) and with the totality of the \(m \) rounds they compute the winning probability \(P_\text{win} \) for the coherence equality game (Equation \ref{eq:interferenceCE}).

From the observed values of  \(P_\text{win} \) and \(d_\epsilon\), for a given security parameter \(\mu\), Alice and Bob compute a key of size $\kappa |\mathbf{D}|$ that is secure against a memory bounded eavesdropper with probability \(1-\mu\), and we show that \(\kappa\) is positive for certain values of \(P_\text{win} \) and \(d_\epsilon\).

\subsection{Protocol description}
\label{sec:setup}

An untrusted source sends a quantum state to Alice's and Bob's labs, and Alice (resp. Bob) randomly chooses to either attempt to detect the particle ($x=1$, resp. $y=1$) or do nothing ($x=0$, resp. $y=0$). Attempting to detect returns an output $\alpha$ (resp. $\beta$) which indicates whether or not the particle was detected. Alice (resp. Bob) then sends their half of the state to the untrusted server $S_A$ (resp. $S_B$), which output outcomes $a$ (resp. $b$), as represented in Figure \ref{fig:setup_QKD}.

The precise steps of the protocol are as follows:
\begin{enumerate}
    \item Take as input parameters $\mu,\eta$. For each round $i\in[m]$, Alice and Bob receive state $\rho_i$ from a source and randomly choose to either detect or do nothing, according to their secret bits $x_i,y_i$ respectively. In case they attempt detection, they receive outcomes \(\alpha_i,\beta_i\), and every round they receive outcomes \(a_i,b_i\) from the untrusted servers.
    \item Alice and Bob share their choices of inputs, $X$ and $Y$ respectively, and the outputs $A,B$ from their respective servers $S_A$ and $S_B$ on a public authenticated channel.
    \item Alice chooses a random set \textbf{B} as a fraction $\gamma$ of the rounds in \textbf{D}, where \textbf{D} is the set of rounds $i$ such that $x_i=y_i=1$, i.e. when both Alice and Bob chose detection, and share the outcomes $\alpha,\beta$ of their detections.
    
    \item From their input choices and from the outputs of $S_A$ and $S_B$, $A$ and $B$ respectively, Alice and Bob calculate the fraction of rounds winning the CE game,
    \begin{equation}
    \hat{P}_\text{win} =  \frac{1}{m}\sum_{i\in[m]}\mathbf{1}_{a_i\oplus b_i=x_i\oplus y_i},
    \label{eq:interferenceCE}
    \end{equation}
    and the fraction of rounds \textit{not} satisfying the single-particle condition, estimated from the public results of rounds in \textbf{B},
    \begin{equation}
    \hat{d_\epsilon}\equiv \frac{1}{|\mathbf{B}|}\sum_{i\in \mathbf{B}}\mathbf{1}_{\alpha_i=\beta_i=1}.
    \label{eq:singledetectionCE}
    \end{equation}
    \item Bob inverts the bits of his set of detection results, $\beta$.\footnote{This step is necessary, since up to here $\alpha$ and $\beta$ are expected to be anti-correlated.} Using the information of rounds in \textbf{B}, Alice and Bob perform  \emph{information reconciliation}. 
    If the fraction of agreement between their results is smaller than $\eta$, they abort the protocol, otherwise, Bob communicates $\ell = H(1.1\eta)|\mathbf{D}|+\log(2/\varepsilon)$  bits of information to Alice.
    
    \item For the security parameter $\mu$, determine $\kappa(\widetilde{P}_\text{win},\Tilde{d_\epsilon},\mu)$ such that, with probability at least $1-\mu$, the protocol is secure and they obtain
    \begin{equation}
        \kappa |\mathbf{D}|-\ell+O(\log 1/\varepsilon)
    \end{equation}
    bits of secret key.
\end{enumerate}

\subsection{Security of the protocol}

In this subsection we give the assumptions made in the security of the protocol, as well as a sketch of the security proof, which is found integrally in the Appendix. We have elected to simply show a sketch of the proof since it follows the standard techniques of security proofs in the SDI-QKD literature, with the major difference being the bound on the min-entropy of Alice's and Bob's results derived from their playing of the CE game with the servers.

\subsubsection{Security assumptions}
\label{assumptions}

The protocol is proven to be secure, when we assume that Eve cannot try to estimate the results in Alice's and Bob's labs after getting partial information about them (i.e. after the round of information reconciliation). This is referred to as the ``bounded quantum storage'' model as it is equivalent to considering that Eve has some space or time bound on her quantum memory forcing her to perform all her measurements before Alice and Bob share any information about their results.

 We do not assume that the devices behave identically and independently in each round, and allow for internal memory that takes into account all previous rounds. Alice and Bob should also be able to ensure that their devices function as particle detectors, following the requirements given in Equation \ref{eq:detectorproperties}, and to verify from timing constraints that the quantum measurements performed by the servers  are spacelike separated in order to respect the non-signalling condition.

 \subsubsection{Security proof sketch}

The security of the protocol is based on the fact that a set of state and measurements $(\rho,M_{A}^a,M_{B}^b)$ picked by the adversary cannot simultaneously be used to win at the CE game and allow for Alice and Bob's  detection outcome bits $\alpha$ and $\beta$ to be completely predictable.
This feature of the game is captured in Proposition \ref{lemma:sdpNoGo}, which places a bound on the information that Eve can obtain on the detection results of Alice and Bob, whenever both decide to attempt detection in their labs.

To apply Proposition \ref{lemma:sdpNoGo}, we use the results obtained by Alice and Bob in their $m$ rounds of the CE game to estimate the behavior of the setup.

In the case of $P_\win$, Alice and Bob can estimate their knowledge of the devices by sharing the entirety of the rounds (values of $x,y,a,b$) and applying the Azuma-Hoeffding inequality \cite{azuma1967}. We show that they can apply their estimation to the detection rounds in Lemma \ref{lemm:boundI}. For the case of $d_\epsilon$, they must use a subset of the detection rounds so as not to lose the entire key. They can still estimate their knowledge using a Chernoff bound and the Azuma-Hoeffding inequality (Lemmas \ref{lemm:F1} and \ref{lemm:boundd}).

From the results of Lemmas \ref{lemm:boundI} and \ref{lemm:boundd}, Alice and Bob obtain values $\widetilde{P}_\win,\widetilde{d_\epsilon}$ bounding the behaviour of whole setup. These two parameters can be used to lower-bound the amount of private information between them, as enunciated in the following theorem, which is proven in full in the Appendix, and along the lines of \cite[Section A.2]{Pironio_2010_Random}.

\begin{restatable*}[Bound on the min-entropy]{thm}{thmboundminentropy}
\label{thm:boundminentropy}

 Let $\mu>0$ be a security parameter and assume that the protocol does not abort, and let $\widehat{P}_\win,\hat{d_\epsilon}$ be the observed values for the CE game and single-particle probabilities. Then there exists a choice of values $\widetilde{P}_\win<\hat{P}_\win$ and $\widetilde{d_\epsilon}>\widehat{d_\epsilon}$ such that, with probability at least $1-\mu$, we have that
\begin{equation*}
    H_{{\min}}(\alpha_\mathbf{D}|ABXY)\geq H_{\min}(\widetilde{P}_\win,\Tilde{d_\epsilon}) |\mathbf{D}|.
\end{equation*}
\end{restatable*}

The main result of the security proof is the result of step 6 of the protocol, i.e.\ the privacy amplification step.

\begin{restatable*}[Privacy amplification]{thm}{thmprivacyamplification}
\label{thm:privacyamplification}
After $m$ rounds, assume that the protocol does not abort and let $\mu:=e^{-c_0 m}>0$ be a security parameter, for $c_0>0$. Then, for any $\varepsilon>0$ with probability at least $1-\mu$, Alice and Bob can perform information reconciliation by sharing $\ell$ bits of information and performing privacy amplification to obtain 
\begin{equation*}
   \kappa|\mathbf{D}|-\ell+O(\log 1/\varepsilon)
\end{equation*}
secure bits of information, where $\kappa$ is a constant that only depends on the values observed for $P_\text{win},d_\epsilon$ and the security parameter $\mu$.
\end{restatable*}

\section{Conclusion}
\label{sec:four}
We introduced a noise-robust generalization of the Coherence Equality game and used SDP techniques in order to compute its optimal quantum bounds as a function of the noise parameter. We were able to identify three distinct classes of coherence resources in the game: non-coherent states, separable coherent states, and entangled coherent states, and computed individually the bounds for all these (see Fig. \ref{fig:Qbounds}). Accordingly, the game can be interpreted as a coherence witness allowing for the certification of the type of coherence resource used and consequently of entanglement, whose values we compute for the optimal bound achieved by the entangled coherent states (see Fig. \ref{fig:VN}).

We further introduced a SDI QKD protocol, based on the Coherence Equality game,  where Alice and Bob need only to implement fixed basis measurements. Our protocol is proven to be unconditionally secure in the quantum-bounded-storage model, and is relevant mainly as a proof-of-concept for the unification of both frameworks of device-independent and semi-quantum key distribution. The novelty relies in using a coherence based game, rather than the usual Bell tests, as a basis for the security of the protocol. This allows for the certification on quantum correlations with fixed single-basis measurements, both for Alice and Bob. The security proof follows closely the standard approach found in \cite{Pironio_2010_Random}, which allows one to establish the security of the protocol, in the quantum-bounded storage model, from the performance in the game. In fact, although we cast the game within a QKD protocol, because the security proof comes from a bound on the guessing probability the game could also be alternatively adapted for random number generation.

\begin{acknowledgements}
M.S. acknowledges the Calouste Gulbenkian Foundation for its scholarship program New Talents in Quantum Technologies, and would like to thank Nikola Paunković for fruitful discussion on the  work of del Santo and Dakić.

R.F. would like to thank Flavio del Santo, and Borivojie Dakić for insightful discussions regarding the nature of the coherence equality game.  RF acknowledges funding from FCT/MCTES through national funds and when applicable EU funds under the project UIDB/50008/2020, and the QuantaGENOMICS project, through the EU H2020 QuantERA II Programme.

E.Z.C. thanks the support from Funda\c c\~ao para a Ci\^encia e a Tecnologia (FCT, Portugal) through project UIDB/50008/2020.

This work was supported in part by the QuantaGENOMICS project, through the EU H2020 QuantERA II Programme, Grant Agreement No 101017733, and by funding organisations, The Foundation for Science and Technology – FCT (QuantERA/0001/2021), Agence Nationale de la Recherche - ANR, and State Research Agency – AEI.

\end{acknowledgements}

\bibliographystyle{quantum}
\bibliography{refs}

\onecolumn\newpage
\appendix

\section{Analytical solution to the noise-robust CE game}
\label{AppendixA}

We take Eq. (\ref{eq:succ_simplified}), and derive it versus $c_{00}$, $c_{01}$, and $n_x$. To find the maximum, we impose that each derivatives must equal zero. We solve the first two for $c_{00}$ and $c_{01}$ respectively, to obtain
\begin{equation}\label{eq:AppA_0}
    c_{00} = -\frac{d_\epsilon^{1/2}n_x^2}{2+n_x^2},
\end{equation}
and
\begin{equation}
    c_{01} = \frac{d_\epsilon^{1/2}n_x\sqrt{1-n_x^2}}{1+n_x^2}.
\end{equation}
We substitute these into the third derivative, and find the following equation for $n_x$,

\begin{equation}\label{eq:AppA_1}
    n_x(4+28d_\epsilon+(12+24d_\epsilon)n_x^2+(13+4d_\epsilon)n_x^4+6n_x^6+n_x^8) = 0
\end{equation}

Using the same substitution for the normalization condition, we obtain another polynomial equation,
\begin{equation}\label{eq:AppA_2}
    4-4d_\epsilon+(12-20d_\epsilon)n_x^2+(13-14d_\epsilon)n_x^4+(6-2d_\epsilon)n_x^6+n_x^8 = 0
\end{equation}

The solution to Eqs. (\ref{eq:AppA_1}) and (\ref{eq:AppA_2}) is an analytical solution\footnote{The solution does not seem to have a compact explicit expression.} to the optimization problem considered in the main text, i.e. the SDP maximizing the noise-robust CE game success probability with a quantum model. Using Eq. (\ref{eq:AppA_0}), one then readily obtains $x(d_\epsilon) \equiv c_{00}^2 (\epsilon)$.

\section{Security proof}

In this section, we show that the QKD protocol is secure given the assumptions in Subsection \ref{assumptions}, with a linear key rate and in the presence of noise. The main result is the following theorem.

\thmprivacyamplification

\begin{proof}
Follows from Theorems \ref{thm:konig2009} and \ref{thm:boundminentropy}, and the result on information reconciliation, Theorem \ref{thm:inforeconfinal}.
\end{proof}

 Since the proof of Theorem \ref{thm:inforeconfinal}, which  pertains to the information reconciliation part of the final result, is not unique to this protocol and follows closely the standard approach found in the literature \cite{rennerphd}, we  will state the theorem and proof it 
 separately in Subsection \ref{subsec:inforec}.
 
Now we start by recalling the relation between the min-entropy and the amount of private information that Alice and Bob can extract by classical communication, expressed in the following theorem.

\begin{thm}[Privacy amplification \cite{Konig_2009}]
Suppose that there is an information reconciliation protocol requiring at most $\ell$ bits of communication. Then, for any $\varepsilon>0$, there is a privacy amplification protocol which extracts
\label{thm:konig2009}
\begin{equation}
    H_{\min}(\alpha_\mathbf{D}|\mathcal{E'})-\ell+O(\log 1/\varepsilon)
\end{equation}
bits of key.
\end{thm}

To establish how the CE game provides a bound on the min-entropy $H_{\min}(\alpha_\mathbf{D}|\mathcal{E}')$, we start by considering a property of each round of the raw key generation, guaranteed by Proposition 1, which we restate once again.

\lemmasdpNoGo

The following results are the steps needed to ensure that Alice and Bob have enough statistical information at the end of the protocol to apply Proposition \ref{lemma:sdpNoGo}, taking into consideration that their devices and the adversary do not necessarily act the same way in every round. In fact, we allow the behavior at round \(i\) to be a function of all inputs and outputs up to round \(i-1\), represented by the variable \(W^i:=(X^{<i},Y^{<i}, A^{<i},B^{<i},\alpha^{<i},\beta^{<i})\).

Using the Azuma-Hoeffding inequality, we see that, for a large number of rounds, the observed behavior of the devices is close to their expected behavior, on average over all rounds.

\begin{lemm}\label{lemm:E1}
Let $\hat{P}_{\textup{win}}$ be the estimated winning probability for the CE game. Then 
\begin{equation*}
    \Pr(\frac{1}{m}\sum_{i=1}^m {P}_{\textup{win}}(W^i)\leq \hat{P}_{\textup{win}}- \varepsilon)\leq \exp{-\frac{m\varepsilon^2}{32(1+\textup{Max}({P}_{\textup{win}}))^2}}.
\end{equation*}
\end{lemm}

\begin{proof}
Similar argument to \cite[Section A.2]{Pironio_2010_Random}.
Consider the random variable
\begin{equation}
    \hat{P}_i=4\times\mathbf{1}_{a_i\oplus b_i=x_i\oplus y_i}.
\end{equation}
Its expectation conditioned on the past $W^i$ is equal to $\mathds{E}(\hat{P}_i|W^i)=P(W^i)$. The observed value for the CE game is $\hat{P}=\frac{1}{m}\sum_{i=1}^m \hat{P}_i$. Consider now the random variable $Z^k=\sum_{i=1}^k(P_i-P(W^i))$. It is true that $(i) \, |Z^k|<\infty$, and that $(ii)\,\mathds{E}(Z^k|W^1,\dots,W^j)=\mathds{E}(Z^k|W^j)=Z^j$, for $j\leq k$. Therefore the sequence $\{Z^k:\,k\geq 1\}$ is a martingale with respect to the sequence $\{W^k:\,k\geq 2 \}$.

The range of the martingale increments is bounded by $|P_i-P(W_i)|\leq 4\,(1+\textup{Max}({P}_{\textup{win}}))$. Applying the Azuma-Hoeffding inequality completes the proof.
\end{proof}

Now, we wish to show that not only is the observed behavior valid on average over all rounds, but that it remains so when we look only at the rounds where Alice and Bob generate the key, i.e. the rounds in \(\mathbf{D}\). Since these rounds are chosen uniformly at random, we can apply a Chernoff bound and see that this is indeed true.

\begin{lemm}\label{lemm:boundI}
Let $\mathbf{D}$ be the set of detection rounds used for the raw key and $\hat{P}_{win}$ the estimated winning probability for the CE game. Then we have that
\begin{align}
    &\Pr(\frac{1}{|\mathbf{D}|}\sum_{i\in\mathbf{D}}P_\text{win}(W^i)\geq (1-\delta)(\hat{P}_{win}-\varepsilon))\\
    &\geq 1- \exp(-m\frac{\delta^2}{8}(\hat{P}_{win}-\varepsilon))-\exp(-\frac{m\varepsilon^2}{32(1+\textup{Max}({P}_{\textup{win}}))^2}).
\end{align}
\end{lemm}

\begin{proof}
Consider the events, for $\delta>0$,
\begin{align}
    &E_1:=\frac{1}{m}\sum_{i=1}^m P_\text{win}(W^i)>\hat{P}_{win}- \varepsilon,\\
    &E_2:=\frac{1}{|\mathbf{D}|}\sum_{i\in\mathbf{D}}P_\text{win}(W^i)\geq (1-\delta)(\hat{P}_{win}-\varepsilon).
\end{align}
We have that $P(E_2)\geq P(E_2\land E_1)=P(E_2|E_1)P(E_1)$. From Lemma \ref{lemm:E1} it follows that $P(E_1)\geq 1- \delta_1$. A bound for $P(E_2|E_1)$ is given by a Chernoff bound,
\begin{align}
    &\Pr(\frac{1}{|\mathbf{D}|}\sum_{i\in\mathbf{D}}P_\text{win}(W^i)\geq (1-\delta)\frac{1}{m}\sum_{i=1}^m P_\text{win}(W^i))\geq 1-\exp{-|\mathbf{D}|\frac{\delta^2}{2}\frac{1}{m}\sum_{i=1}^m P_\text{win}(W^i)}.
\end{align}
Conditioning on $E_1$, we can write 
\begin{align}
    \frac{1}{|\mathbf{D}|}\sum_{i\in\mathbf{D}}P_\text{win}(W^i|E_1)&\geq (1-\delta)\frac{1}{m}\sum_{i=1}^m P_\text{win}(W^i|E_1)\\
    &>(1-\delta)(\hat{P}_{win}-\varepsilon)
\end{align}
and therefore
\begin{align}
\Pr(\frac{1}{|\mathbf{D}|}\sum_{i\in\mathbf{D}}P_\text{win}(W^i)\geq (1-\delta)(\hat{P}_{win}-\varepsilon)\bigg|E_1)&\geq 1-\exp{-|\mathbf{D}|\frac{\delta^2}{2}\frac{1}{m}\sum_{i=1}^m P_\text{win}(W^i|E_1)}\\
&> 1-\exp{-|\mathbf{D}|\frac{\delta^2}{2}(\hat{P}_{win}-\varepsilon)},
\end{align}
which concludes the proof.
\end{proof}

Since our upper bound is a function also of the single-detection probability, we must carry out a similar analysis over \(\hat{d_\epsilon}\). There is a small distinction at the end which is that we cannot use the information of all the rounds to estimate the single-detection probability, since that would simply release the entire raw key. Instead, we sacrifice a fraction \(\gamma\) chosen randomly from the key, achieve similar conclusions about the rounds \(\mathbf{B}\) and then use a Chernoff bound in relation to the full detection set \(\mathbf{D}\).

\begin{lemm}\label{lemm:F1}
Let $\hat{d_\epsilon}$ be the observed value for the single-detection condition, taking the detection results of rounds in $\mathbf{B}$. Then 
\begin{equation}
    \Pr(\bigg|\frac{1}{|\mathbf{B}|}\sum_{i\in\mathbf{B}} d_\epsilon(W^i)-\hat{d_\epsilon}\bigg|\geq \varepsilon)\leq 2 \exp{-\frac{|\mathbf{B}|\varepsilon^2}{8}}.
\end{equation}
\end{lemm}
\begin{proof}
Comes directly from applying the Azuma-Hoeffding inequality.
\end{proof}

\begin{lemm}\label{lemm:boundd}
Let $\mathbf{B}$ be the subset of detection rounds $\mathbf{D}$ used in estimating $\hat{d_\epsilon}$, such that $|\mathbf{B}|=\gamma|\mathbf{D}|=\gamma m/4$. Then, for $\varepsilon,\delta>0$,
\begin{align}
    &\Pr(\frac{\hat{d_\epsilon}+\varepsilon}{1-\delta}>\frac{1}{|\mathbf{D}|}\sum_{i\in\mathbf{D}}d_\epsilon(W^i))\\
    &\geq 1-\exp{-\frac{\delta^2}{8}\gamma^2 m (\hat{d_\epsilon}-\varepsilon)}-2\exp{-\frac{\gamma m\varepsilon^2}{32}}.
\end{align}
\end{lemm}

\begin{proof}
Same argument as Lemma \ref{lemm:boundI}. Consider the events 
\begin{align}
    &F_1:=\bigg|\frac{1}{|\mathbf{B}|}\sum_{i\in\mathbf{B}} d_\epsilon(W^i)-\hat{d_\epsilon}\bigg|< \varepsilon,\\
    &F_2:=\frac{\hat{d_\epsilon}+\varepsilon}{1-\delta}>\frac{1}{|\mathbf{D}|}\sum_{i\in\mathbf{D}}d_\epsilon(W^i).
\end{align}
The probability $P(F_1)\geq 1-\delta_1$ is given by Lemma \ref{lemm:F1}. Conditioning on $F_1$, the probability $P(F_2|F_1)$ is given by a Chernoff bound:
\begin{align}
    \Pr(\frac{1}{|\mathbf{B}|}\sum_{i\in\mathbf{B}}d_\epsilon(W^i)>(1-\delta)\frac{1}{|\mathbf{D}|}\sum_{i\in\mathbf{D}}d_\epsilon(W^i)\bigg|F_1)&\geq 1-\exp{-\frac{\delta^2}{2}\frac{|\mathbf{B}|}{|\mathbf{D}|}\sum_{i\in\mathbf{D}}d_\epsilon(W^i|F_1)}\\
    &\geq 1-\exp{-\frac{\delta^2}{2}\frac{|\mathbf{B}|^2}{|\mathbf{D}|}(\hat{d_\epsilon}-\varepsilon)}.
\end{align}
The last inequality is obtained by noting that
$\sum_{i\in\mathbf{D}}d_\epsilon(W^i|F_1)\geq\sum_{i\in\mathbf{B}}d_\epsilon(W^i|F_1)>(\hat{d_\epsilon}-\varepsilon)|\mathbf{B}|$.
\end{proof}

We are now in position to apply Proposition \ref{lemma:sdpNoGo}. Knowing, on average over \(\mathbf{D}\), the expected behavior of the devices, Alice and Bob can calculate a lower bound for a linear key rate.

\thmboundminentropy

\begin{proof}
This proof follows along the lines of Section A.2 of \cite{Pironio_2010_Random}. Recalling the relation between min-entropy and guessing probability \cite{Konig_2009},
\begin{equation}
    H_{\min}(X|Y)=-\log_2 P_{g}(X|Y).
\end{equation}

We are interested in the min-entropy $H_{\min}(\alpha_\mathbf{D}|ABXY)$ of the detection results $\alpha$  of rounds in $\mathbf{D}$, given that Eve has access to the strings $X,Y,A,B$ of inputs and outputs. Consider the strings $\alpha^d=(\alpha_i)_{i\in [1..d]}$ where $d$ runs through the indices in the set $\mathbf{D}$. Similarly, $a^m,b^m,x^m,y^m$ where $m$ runs through all rounds. We have that
\begin{align} \label{eq1}
-\log_2 P(\alpha^d|a^m b^m x^m y^m)&= -\log_2 \prod_{i\in \mathbf{D}} p(\alpha^i|a^{i-1}b^{i-1}x^{i-1}y^{i-1}\alpha^{i-1}) \\
&= -\log_2 \prod_{i\in \mathbf{D}} p(\alpha^i|W^i) \\
&= \sum_{i \in \mathbf{D}} -\log_2 p(\alpha^i|W^i).
\end{align}
We can apply Proposition \ref{lemma:sdpNoGo} to each of the rounds, and we obtain a constraint (see Figure \ref{fig:Hmin})
\begin{align}
-\log_2 P(\alpha^i|W^i)&>-\log_2 P_{g}(P_\text{win}(W^i),d_\epsilon(W^i))\\
&= H_{\min}(P_\text{win}(W^i),d_\epsilon(W^i)).
\end{align}This bound is true conditioned on any measurement of an eavesdropper before Alice and Bob share any information about their inputs and outputs, since any outcome of measurement in that case amounts to the preparation of a state to be used by Alice and Bob, and the bound is independent of the state being used. Therefore, in this step we assume that the adversary has a bounded quantum memory and cannot delay her measurements so that they are made after the parameter estimation step. With this caveat in mind, we can write
\begin{align}
-\log_2 p(\alpha^d|a^m b^m x^m y^m)&\geq \sum_{i\in \mathbf{D}} H_{\min}(P_\text{win}(W^i),d_\epsilon(W^i))\\
&\geq |\mathbf{D}|\,H_{\min}\bigg(\frac{1}{|\mathbf{D}|}\sum_{i\in \mathbf{D}}P_\text{win}(W^i),\frac{1}{|\mathbf{D}|}\sum_{i\in \mathbf{D}}d_\epsilon(W^i)\bigg).
\end{align}
The last inequality is deduced using the convexity of $f$ and Jensen's inequality for two variables.

Let $\mu>0$ be the a security parameter. Using Lemmas \ref{lemm:boundI} and \ref{lemm:boundd}, by sharing a fraction $\gamma>0$ of the results of detection rounds, for $m$ large enough, we can establish values $\varepsilon,\delta,\varepsilon',\delta'>0$ such that, with probability at least $1-\mu$,
\begin{align}
    &\frac{1}{|\mathbf{D}|}\sum_{i\in\mathbf{D}}d_\epsilon(W^i)<\frac{\hat{d_\epsilon}+\varepsilon}{1-\delta}=:\Tilde{d_\epsilon}\\
    &\frac{1}{|\mathbf{D}|}\sum_{i\in\mathbf{D}}P_\text{win}(W^i)\geq (1-\delta')(\hat{P}_{win}-\varepsilon')=:\Tilde{P}_{win}.
\end{align}
Since $f(P_\text{win},d_\epsilon)$ is increasing with $I$ and decreasing with $d_\epsilon$, we have that
\begin{align}
    H_{\min}(\alpha_\mathbf{D}|ABXY)&=-\log_2 P(\alpha^d|a^m b^m x^m y^m )\geq |\mathbf{D}|\,H_{\min}(\Tilde{P}_{win},\Tilde{d_\epsilon})\equiv |\mathbf{D}|\,\kappa,
\end{align}
which concludes the proof.
\end{proof}

\subsection{Max-entropy bound and information reconciliation}\label{subsec:inforec}

Either from eavesdropping or simply noise in the channels, Alice and Bob's raw keys will not necessarily match, and so they must reconcile their keys without making them public. By sharing a small fraction at random, they are able to bound the max-entropy between them and with high probability obtain an identical key. 

\begin{lemm}[Bound on the max-entropy]\label{thm:inforecobound}
Suppose Alice and Bob do not abort the protocol after Step 4. Let \textbf{D} be the set of detection rounds. Then, with probability at least $1-2e^{-\gamma\,\eta\,|\mathbf{D}|/250}$,
\begin{equation*}
    H_{\max}(\beta_{\mathbf{D}}|\alpha_{\mathbf{D}})\leq H(1.1\eta)|\mathbf{D}|,
\end{equation*}
where $H(\cdot)$ is the typical binary entropy.
\end{lemm}

\begin{proof}
Let $X=\gamma |\mathbf{D}|\eta_\mathbf{B}$ be the observed error rate, and $\mu=\gamma |\mathbf{D}|\eta_\mathbf{D}$ the expected value. Since the protocol aborts if $\eta_B>\eta$, we have that $\Pr(1.1\eta <\eta_\mathbf{D} \land \neg\textrm{aborts})=\Pr(1.1\eta<\eta_\mathbf{D}\land \eta_\mathbf{B}\leq\eta)\leq P(1.1\eta_\mathbf{B}<\eta_\mathbf{D}\land 1.1\eta<\eta_\mathbf{D})$. Using a Chernoff bound, we have that
\begin{align}
\Pr(X<\frac{1}{1.1}\mu\land 1.1\eta<\eta_\mathbf{D})
&\leq\Pr(X<\frac{1}{1.1}\mu\bigg| 1.1\eta<\eta_\mathbf{D})\\
&\leq \exp{-\frac{1}{242}\gamma|\mathbf{D}|\eta_\mathbf{D}}\\
&\leq\exp{-\frac{1}{250}\gamma |\mathbf{D}|\eta}.
\end{align}
Therefore, with probability at least $1-\exp{-\frac{1}{250}\gamma|\mathbf{D}|\eta}$, the noise rate of $\mathbf{D}$ is at most $1.1\eta$. This implies that, for a fixed $\alpha_\mathbf{D}$, there are at most $2^{|\mathbf{D}|H(1.1\eta)}$ possible values for $\beta_\mathbf{D}$. Using Remark 3.1.4 in \cite{rennerphd}, we have that $H_{\max}(\beta_\mathbf{D}|\alpha_\mathbf{D})\leq H(1.1\eta)|\mathbf{D}|$.
\end{proof}

\begin{thm}[Information reconciliation]\label{thm:inforeconfinal}
Suppose Alice and Bob do not abort the protocol after Step 4. Then, for any $\varepsilon>0$, they can perform information reconciliation on their bit strings $\alpha_\mathbf{D},\beta_\mathbf{D}$ sacrificing at most the following amount of bits of information
\begin{equation*}
\ell\leq H(1.1\eta)|\mathbf{D}|+\log(2/\varepsilon).
\end{equation*}
\end{thm}
\begin{proof}
Follows from Lemma 6.3.3 in \cite{rennerphd} and Lemma \ref{thm:inforecobound}.
\end{proof}

\end{document}